\documentclass[twoside,leqno,onecolumn]{article}

\usepackage{ltexpprt}
\usepackage{srcltx}
\usepackage{color}
\usepackage{graphicx}
\usepackage{bm}
\usepackage[cmex10]{amsmath}
\usepackage{algorithmic}
\usepackage{algorithm}
\usepackage{amssymb}
\usepackage{multirow}

\newtheorem{assumption}{Assumption}
\newtheorem{remark}{Remark}

\begin{document}

\title{Extremum Seeking-based Indirect Adaptive Control for Nonlinear Systems with State and Time-Dependent Uncertainties}
\author{Mouhacine Benosman and Meng Xia\thanks{Mouhacine Benosman
(m{\_}benosman@ieee.org) is with Mitsubishi Electric Research
Laboratories, Cambridge, MA 02139, USA.Meng Xia (mxia@nd.edu) is
with the with the Department of Electrical Engineering, University
of Notre Dame, Notre Dame, IN 46556, USA (collaborated on this
project during her internship at MERL). }\\\\This manuscript is
based on the recent results published by the authors in\\ the IEEE
European Control Conference 2015 and the SIAM Control and its
Applications 20015}

\date{}
\maketitle

\begin{abstract}
We study in this paper the problem of adaptive trajectory tracking
for nonlinear systems affine in the control with bounded
state-dependent and time-dependent uncertainties. We propose to
use a modular approach, in the sense that we first design a robust
nonlinear state feedback which renders the closed loop input to
state stable (ISS) between an estimation error of the uncertain
parameters and an output tracking error. Next, we complement this
robust ISS controller with a model-free multiparametric extremum
seeking (MES) algorithm to estimate the model uncertainties. The
combination of the ISS feedback and the MES algorithm gives an
indirect adaptive controller. We show the efficiency of this
approach on a two-link robot manipulator example.
\end{abstract}

\section{Introduction}

Input-output feedback linearization has been proven to be a
powerful control design for trajectory tracking and stabilization
of nonlinear systems \cite{I89}. The basic idea is to first
transform a nonlinear system into a simplified linear equivalent
system and then use the linear design techniques to design
controllers in order to satisfy stability and performance
requirements. One shortcoming of the feedback linearization
approach is that it requires precise system modelling \cite{I89}.
When there exist model uncertainties, a robust input-output
linearization approach needs to be developed. For instance,
high-gain observers \cite{Frei08} and linear robust controllers
\cite{Fran06} have been proposed in combination with the feedback
linearization techniques. Another approach to deal with model
uncertainties is using adaptive control methods. Of particular
interest to us is the modular approach to adaptive nonlinear
control, e.g. \cite{KKK95}. In this approach, first the controller
is designed by assuming all the parameters are known and then an
identifier is used to guarantee certain boundedness of the
estimation error. The identifier is independent of the designed
controller and thus this is called `modular' approach.

On the other hand, extremum seeking (ES) method is a model-free
control approach, e.g.\cite{Ariy03}, which has been applied to
many industrial systems, such as electromagnetic actuators \cite{
Beno13,Atinc13}, compressors \cite{WYK00}, and stirred-tank
bioreactors \cite{Zhang03}. Many papers have dedicated to
analyzing the ES algorithms convergence when applied to a static
or a dynamic known maps, e.g.\cite{Rote00,Ariy03,KW00,Teel01},
however, much fewer papers have been dealing with the use of ES in
the case of static or dynamic uncertain maps. The case of ES
applied to an uncertain static and dynamic mapping, was
investigated in \cite{Nesic13}, where the authors considered
systems with constant parameter uncertainties. However, in
\cite{Nesic13}, the authors used ES to optimize a given
performance (via optimizing a given performance function), and
complemented the ES algorithm with classical model-based
filters/estimators to estimate the states and the unknown constant
parameters of the system, which is one of the main differences
with the approach that we want to present here (originally
introduced by the authors for a specific mechatronics application
in \cite{Atinc13,BA13,patent_1}), where the ES is not only used to
optimize a given performance but is also used to estimate the
uncertainties of the model, without the need for extra model-based
filters/estimators. \\In this work, we build upon the existing ES
results to provide a framework which combines ES results and
robust model-based nonlinear control to propose an ES-based
indirect adaptive controller, where the ES algorithm is used to
estimate, in closed-loop, the uncertainties of the model.
Furthermore, we focus here on a particular class of nonlinear
systems which are input-output linearizable through static state
feedback \cite{Khalil02}. We assume that the uncertainties in the
linearized model are bounded additive as guaranteed by the
`matching condition' \cite{Beno09}. The control objective is to
achieve asymptotic tracking of a desired trajectory. The proposed
adaptive control is designed as follows. In the first step, we
design a controller for the nominal model (i.e. when the
uncertainties are assumed to be zero) so that the tracking error
dynamics is asymptotically stable. In the second step, we use a
Lyapunov reconstruction method \cite{Beno10} to show that the
error dynamics are input-to-state stable (ISS)
\cite{Khalil02,Sontag95} where the estimation error in the
parameters is the input to the system and the tracking error
represents the system state. Finally, we use ES to estimate the
uncertain model parameters so that the the tracking error will be
bounded and decreasing, as guaranteed by the ISS property. To
validate the results, we apply our results on a two-link robotic
manipulators \cite{Spong92}.\\ Similar ideas of ES-based adaptive
control for nonlinear systems have been introduced in
\cite{Beno13,Atinc13}. In these two works, the problem of adaptive
robust control of electromagnetic actuators was studied, where ES
was used to tune the feedback gains of the nonlinear controller in
\cite{Beno13} and ES was used to estimate the unknown parameters
in \cite{Atinc13}. An extension to the general case of nonlinear
systems was proposed in \cite{ben_ifac14,ben_ECC14}. We relax here
the strong assumption, used in \cite{ben_ifac14,ben_ECC14}, about
the existence of an ISS feedback controller, and propose a
constructive proof to design such an ISS feedback for the
particular case of nonlinear systems affine in the control.

The rest of the paper is organized as follows. In Section
\ref{sec:pre}, we present notations, and some fundamental
definitions and results that will be needed in the sequel. In
Section \ref{sec:problem}, we provide our problem formulation. The
nominal controller design are presented in Section \ref{sec:norm}.
In Section \ref{sec:robust}, a robust controller is designed which
guarantees ISS from the estimation errors input to the tracking
errors state. In Section \ref{sec:ES}, the ISS controller is
complemented with an MES algorithm to estimate the model
uncertainties. Section \ref{sec:sim} is dedicated to an
application example and the paper conclusion is given in Section
\ref{sec:con}.

\section{Preliminaries}
\label{sec:pre} Throughout the paper, we use $\|\cdot\|$ to denote
the Euclidean norm; i.e. for a vector $x\in\mathbb{R}^n$, we have
$\|x\| \triangleq\|x\|_2 =\sqrt{x^Tx}$, where $x^T$ denotes the
transpose of the vector $x$.  The $1$-norm of $x\in\mathbb{R}^n$
is denoted by $\|x\|_1$. We use the following norm properties for
the need of our proof:
\begin{enumerate}
\item for any $x\in\mathbb{R}^n$, $\|x\|\leq\|x\|_1$;
\item for any $x,~y\in\mathbb{R}^n$, $\|x\|-\|y\|~\leq~\|x-y\|$;
\item for any $x,~y\in\mathbb{R}^n$, $x^Ty~\leq~\|x\|\|y\|$.
\end{enumerate}
Given $x\in\mathbb{R}^m$, the signum function is defined as
\begin{align*}
{\rm{sign} } (x)~\triangleq~[{\rm{sign} } (x_1),~{\rm{sign} } (x_2),~\cdots,~{\rm{sign} } (x_m)]^T,
\end{align*}
where $x_i$ denotes the $i$-th ($1\leq i\leq m$) element of $x$ and
\[ {\rm{sign} } (x_i) ~=~ \left\{ \begin{array}{lll}
         1 & \quad\mbox{if $x_i ~>~ 0$}\\
         0 & \quad\mbox{if $x_i ~=~ 0$}\\
        -1 & \quad\mbox{if $x_i ~<~ 0$}\end{array} \right. \]
We have $x^T{\rm{sign} } (x)~=~\|x\|_1$.

For an $n\times n$
matrix $P$, we denote by $P>0$ if it is positive definite. Similarly, we denote
by $P<0$ if it is negative definite. We use ${\rm diag}\{A_1,~A_2,~\cdots,~A_n\}$ to denote a
diagonal block matrix with $n$ blocks. For a matrix $B$, we denote $B(i,j)$ as the element that locates at the $i$-th row and $j$-th column of matrix $B$. We denote $I_n$ as the identity matrix or simply $I$ if the dimension is clear from the context.

  We use $\dot{f}$ to denote the time derivative of $f$ and $f^{(r)}(t)$ for the $r$-th derivative of $f(t)$, i.e. $f^{(r)}~\triangleq~\frac{d^r f}{dt}$. We denote by $\mathbb{C}^k$ functions that are $k$ times differentiable and by $\mathbb{C}^\infty$ a smooth function. A continuous function $\alpha : [0,a)\rightarrow [0,\infty)$ is said to belong to class $\mathcal{K}$ if it is strictly increasing and $\alpha(0) = 0$. It is said to belong to class $\mathcal{K}_{\infty}$ if $a = \infty$ and $\alpha(r)\rightarrow \infty$ as $r\rightarrow \infty$ \cite{Khalil02}. A continuous function $\beta:[0,a)\times[0,\infty)\rightarrow [0,\infty)$ is said to belong to class $\mathcal{KL}$ if, for a fixed
$s$, the mapping $\beta(r,s)$ belongs to class $\mathcal{K}$ with respect to $r$ and, for each fixed $r$, the mapping $\beta(r,s)$ is decreasing with respect to $s$ and $\beta(r,s)\rightarrow 0$ as $s\rightarrow \infty$ \cite{Khalil02}.
%

Consider the system
\begin{align}
\label{eq:sys}
\dot{x} ~=~f(t,x,u)
\end{align}
where $f: [0,\infty)\times \mathbb{R}^n\times\mathbb{R}^m\rightarrow \mathbb{R}^n$ is piecewise continuous in $t$ and locally Lipschitz in $x$ and $u$, uniformly in $t$. The input $u(t)$ is piecewise continuous, bounded function of $t$ for all $t\geq0$.

\begin{Definition}[\cite{Khalil02,Mali05}]
\label{def:ISS}
The system (\ref{eq:sys}) is said to be \emph{input-to-sate stable} (ISS) if there exist a class $\mathcal{KL}$ function $\beta$ and a class $\mathcal{K}$ function $\gamma$ such that for any initial state $x(t_0)$ and any bounded input $u(t)$, the solution $x(t)$ exists for all $t\geq t_0$ and satisfies
\begin{align*}
\|x(t)\|\leq \beta(\|x(t_0)\|,t-t_0) + \gamma(\sup_{t_0\leq\tau\leq t} \|u(\tau)\|).
\end{align*}
\end{Definition}

\begin{theorem}[\cite{Khalil02,Mali05}]
\label{thm:iss}
Let $V: [0,\infty)\times\mathbb{R}^n\rightarrow \mathbb{R}$ be a continuously differentiable function such that
\begin{align}
\label{eq:iss}
\alpha_1(\|x\|) \leq& V(t,x)\leq \alpha_2(\|x\|)\notag\\
\frac{\partial V}{\partial t} + \frac{\partial V}{\partial x}f(t,x,u)&\leq -W(x),\quad \forall \|x\|\geq\rho(\|u\|)>0
\end{align}
for all $(t,x,u)\in[0,\infty)\times\mathbb{R}^n\times \mathbb{R}^m$, where $\alpha_1$, $\alpha_2$ are class $\mathcal{K}_{\infty}$ functions, $\rho$
is a class $\mathcal{K}$ function, and $W(x)$ is a continuous positive definite function on $\mathbb{R}^n$. Then, the system (\ref{eq:sys}) is input-to-state stable (ISS).
\end{theorem}

\begin{remark}
Note that other equivalent definitions for ISS have been given in \cite[pp. 1974-1975]{Mali05}. For instance, Theorem \ref{thm:iss} holds with all the assumptions are the same except that the inequality (\ref{eq:iss}) is replaced by
\begin{align*}
\frac{\partial V}{\partial t} + \frac{\partial V}{\partial x}f(t,x,u)&\leq -\mu(\|x\|) + \Omega(\|u\|)
\end{align*}
where $\mu\in\mathcal{K}_{\infty}\bigcap C^1$ and $\Omega\in\mathcal{K}_{\infty}$.
\end{remark}


\section{The case of nonlinear system model with state dependent
uncertainties}\label{sec:problem} We consider here affine
uncertain nonlinear systems of the form:
\begin{equation}
\begin{array}{l}
\displaystyle \dot{x} ~=~ f(x) + \Delta f(x) + g(x)u,\quad x(0) = x_0\\
\displaystyle y~=~ h(x)
\end{array}
\label{eq:model}
\end{equation}
where $x\in\mathbb{R}^n$, $u\in \mathbb{R}^p$, $y\in\mathbb{R}^m$
($p~\geq~ m$), represent respectively the state, the input and the
controlled output vectors, $x_0$ is a given initial condition,
$\Delta f(x)$ is a vector field representing additive model
uncertainties. The vector fields $f$, $\Delta f$, columns of $g$
and function $h$ satisfy the following assumptions.

\begin{assumption}
The function $f:\mathbb{R}^n \rightarrow \mathbb{R}^n$ and the
columns of $g: \mathbb{R}^n \rightarrow \mathbb{R}^p$ are
$\mathbb{C}^{\infty}$ vector fields on a bounded set $X$ of
$\mathbb{R}^n$ and $h : \mathbb{R}^n \rightarrow \mathbb{R}^m$ is
a $\mathbb{C}^{\infty}$ vector on $X$. The vector field $\Delta
f(x) $ is $\mathbb{C}^1$ on $X$.
\end{assumption}

\begin{assumption}
System (\ref{eq:model}) has a well-defined (vector) relative
degree $\{r_1,~r_2,~\cdots,~r_m\}$ at each point $x^0\in X$, and
the system is linearizable, i.e. $\sum_{i=1}^m r_i=n$.
\end{assumption}


\begin{assumption}
The desired output trajectories $y_{id}$ ($1\leq i \leq m$) are
smooth functions of time, relating desired initial points
$y_{id}(0)$ at $t=0$ to desired final points $y_{id}(t_f)$ at
$t=t_f$.
\end{assumption}

\subsection{Control objectives}
Our objective is to design a state feedback adaptive controller so
that the tracking error is uniformly bounded, whereas the tracking
upper-bound can be made smaller over the ES learning iterations.
We stress here that the goal of the ES is not stabilization but
rather performance optimization, i.e. estimating online the
uncertain part of the model and thus improving the performance of
the overall controller. To achieve this control objective, we
proceed as follows. First, we design a robust controller which can
guarantee the input-to-state stability (ISS) of the tracking error
dynamics w.r.t the estimation errors input. Then, we combine this
controller with a model-free extremum-seeking algorithm to
iteratively estimate the uncertain parameters, to optimize online
a desired performance cost function. \subsection{Nominal
Controller Design} \label{sec:norm}
Under Assumption 2 and nominal conditions, i.e. when $\Delta f(x)
= 0$, system (\ref{eq:model}) can be written as
\begin{align}
\label{eq:linear}
y^{(r)}(t) ~=~ b(\xi(t)) + A (\xi(t)) u(t),
\end{align}
where
\begin{equation}
\begin{array}{l}
\displaystyle y^{(r)}(t) ~=~ [y^{(r_1)}_{1}(t), ~y^{(r_2)}_{2}(t),~\cdots,~ y^{(r_m)}_{m}(t)]^T\\
\displaystyle \xi(t) ~=~ [\xi^1(t),~\cdots,~\xi^m(t)]^T\\
\displaystyle \xi ^i(t) ~=~[y_i(t),~\cdots,~y_i^{(r_i-1)}(t)], \quad 1\leq i\leq m
\end{array}
\label{eq:linear_c}
\end{equation}
The functions $b(\xi)$, $A(\xi)$ can be written as functions of
$f$, $g$ and $h$, and $A(\xi)$ is non-singular in $\tilde{X}$,
where $\tilde X$ is the image of the set of $X$ by the
diffeomorphism $x\rightarrow\xi$ between the states of system
(\ref{eq:model}) and the linearized model (\ref{eq:linear}).

At this point, we introduce one more assumption on system (\ref{eq:model}).

\begin{assumption}
The additive uncertainties $\Delta f(x)$ in (\ref{eq:model})
appear as additive uncertainties in the input-output linearized
model (\ref{eq:linear})-(\ref{eq:linear_c}) as follows (see also
\cite{Beno09})
\begin{align}
\label{eq:model_un} y^{(r)}(t) ~=~ b(\xi(t)) + A (\xi(t)) u(t) +
\Delta b(\xi(t)),
\end{align}
where $\Delta b(\xi(t))$ is $\mathbb{C}^1$ on $\tilde X$.
\end{assumption}

If we consider the nominal model (\ref{eq:linear})first, then we can define a virtual input vector $v(t)$ as
\begin{align}
\label{eq:input}
v(t) ~=~ b(\xi(t)) + A (\xi(t)) u(t).
\end{align}
Combining (\ref{eq:linear}) and (\ref{eq:input}), we can obtain the following input-output mapping
\begin{align}
\label{eq:virtual}
y^{(r)}(t) ~=~v(t).
\end{align}
Based on the linear system (\ref{eq:virtual}), it is straightforward to apply a stabilizing controller for the nominal system (\ref{eq:linear}) as
\begin{align}
\label{eq:norm}
u_{n} = A^{-1}(\xi) \left[v_{s}(t,\xi) - b(\xi)\right],
\end{align}
where $v_s$ is a $m\times1$ vector and the $i$-th ($1\leq i \leq m$) element $v_{si}$ is given by
\begin{align}
\label{eq:feedback}
v_{si} = y_{id}^{(r_i)} - K_{r_i}^{i}(y_i^{(r_i-1)}-y_{id}^{(r_i-1)})-\cdots-K_{1}^{i}(y_i-y_{id}).
\end{align}
Denote the tracking error as $e_i(t) ~\triangleq~ y_i(t)- y_{id}(t)$, we obtain  the following tracking error dynamics
\begin{align}
\label{eq:error}
e_i^{(r_i)}(t) +  K^i_{r_i}e^{(r_i-1)}(t) + \cdots + K^i_1e_i(t) = 0,
\end{align}
where $i \in \{1,~2,~\cdots,~m\}$. By selecting the gains $K^i_j$
where $i \in \{1,~2,~\cdots,~m\}$ and $j\in\{1,~2,~\cdots,~r_i\}$,
we can obtain global asymptotic stability of the tracking errors
$e_i(t)$. To formalize this condition, we make the following
assumption.

\begin{assumption}
There exists a non-empty set $\mathcal{A}$ where
$K^i_j\in\mathcal{A} $ such that the polynomials in
(\ref{eq:error}) are Hurwitz, where $i \in \{1,~2,~\cdots,~m\}$
and $j\in\{1,~2,~\cdots,~r_i\}$.
\end{assumption}


To this end, we define $z=[z^1,~z^2,~\cdots,~z^m]^T$, where $z^i=[e_i,~\dot{e_i},~\cdots,~e_i^{(r_i-1)}]$ and $i\in\{1,~2,\cdots,~m\}$. Then, from (\ref{eq:error}), we can obtain
\begin{align*}
\dot{z}~=~ \tilde{A} z,
\end{align*}
where $\tilde{A}\in \mathbb{R}^{n\times n}$ is a diagonal block matrix given by
\begin{align}
\label{eq:tildA}
\tilde{A}~=~{\rm{diag}}\{\tilde A_1,~\tilde A_2,~\cdots,~\tilde A_m\},
\end{align}
and $\tilde A_i$ ($1\leq i\leq m$) is a $r_i\times r_i$ matrix given by
\[ \tilde{A}_i~=~\left[ \begin{array}{cccccc}
0& 1&  \\
0& &1&\\
0&&&\ddots\\
\vdots&&&&1\\
-K^i_1\quad&-K^i_2\quad&\cdots\quad&\cdots\quad&-K^i_{r_i}
\end{array} \right].\]
As discussed above, the gains $K_j^i$ can be chosen so that the
matrix $\tilde{A}$ is Hurwitz. Thus, there exists a positive
definite matrix $P~>~ 0$ such that (see e.g. \cite{Khalil02})
\begin{align}
\label{eq:lyap}
\tilde A^T P + P \tilde A ~=~-I.
\end{align} 


\subsection{Robust Controller Design} \label{sec:robust} We now
consider the uncertain model (\ref{eq:model}), i.e. when $\Delta
f(x)\neq0$. The corresponding linearized model is given by
(\ref{eq:model_un}) where $\Delta b (\xi(t))\neq0$. The global
asymptotic stability of the error dynamics (\ref{eq:error}) cannot
be guaranteed anymore due to the additive uncertainty $\Delta
b(\xi(t))$. We use Lyapunov reconstruction techniques to design a
new controller so that the tracking error is guaranteed to be
bounded given that the estimate error of $\Delta b(\xi(t))$ is
bounded. The new controller for the uncertain model
(\ref{eq:model_un}) is defined as
\begin{align}
\label{eq:full}
u_{f}~=~u_{n} + u_{r},
\end{align}
where the nominal controller $u_{n}$ is given by (\ref{eq:norm})
and the robust controller $u_r$ will be given later on based on
particular forms of the uncertainty $\Delta b(\xi(t))$. By using
the controller (\ref{eq:full}), from (\ref{eq:model_un}) we obtain
\begin{align}
\label{eq:io_new}
y^{(r)}(t) &~=~ b(\xi(t)) + A (\xi(t)) u_{f} + \Delta b(\xi(t))\notag\\
&~=~ b(\xi(t)) + A (\xi(t)) u_{n} + A (\xi(t)) u_{r}+ \Delta b(\xi(t))\notag\\
& ~=~v_{s}(t,\xi)+ A (\xi(t)) u_{r}+ \Delta b(\xi(t)),
\end{align}
where (\ref{eq:io_new}) holds from (\ref{eq:norm}). Thus, we have
\begin{align}
\label{eq:error_new}
&~e_i^{(r_i)}(t) +  K^i_{r_i}e^{(r_i-1)}(t) + \cdots + K^i_1e_i(t)\notag\\
~=&~A (\xi(t)) u_{r}+ \Delta b(\xi(t))
\end{align}
Further, the dynamics for $z$ is given by
\begin{align}
\label{eq:dyn_z}
\dot{z}~=~ \tilde{A} z + \tilde{B} \delta,
\end{align}
where $\tilde{A}$ is defined in (\ref{eq:tildA}), $\delta$ is a $m\times 1$ vector given by
\begin{align}
\label{eq:delta} \delta~=~A (\xi(t)) u_{r}+ \Delta b(\xi(t)),
\end{align}
and the matrix $\tilde B\in\mathbb{R}^{n\times m}$ is given by
\begin{align}
\label{eq:tildB}
\tilde{B}~=~\left[ \begin{array}{c}
\tilde{B}_1 \\
\tilde{B}_2\\
\vdots\\
\tilde{B}_m
\end{array} \right],
\end{align}
with $\tilde B_i$ ($1\leq i\leq m$) given by a $r_i\times m$ matrix such that
\[ \tilde{B}_i(l,q) ~=~ \left\{ \begin{array}{lll}
         1 & \quad\mbox{if $l = r_i$ and $q = i$}\\
         0  & \quad\mbox{otherwise}\end{array} \right. \]
%
%
If we apply $V(z)=z^TPz$ as a Lyapunov function for the dynamics
(\ref{eq:dyn_z}), where $P$ is the solution of the Lyapunov
equation (\ref{eq:lyap}), then we obtain
\begin{align}
\label{eq:general}
\dot{V}(t)&~=~\frac{\partial V}{\partial z}\dot{z}\notag\\
&~=~z^T(\tilde A^T P + P \tilde A)z + 2 z^T P \tilde B \delta\notag\\
&~=~-\|z\|^2 + 2 z^T P \tilde B \delta,
\end{align}
%
where $\delta$ given by (\ref{eq:delta}) depends on the robust controller $u_r$.

Next, we will design the controller $u_{r}$ based on the
particular forms of the uncertainties that appear in
(\ref{eq:model_un}), i.e. $\Delta b(\xi(t))$. For notational
convenience, the unknown parameter vector/matrix is denoted by
$\Delta$ and the estimate for the unknowns is denoted by
$\widehat{\Delta}(t)$. Further, the estimation error vector/matrix
is given by $e_\Delta(t) = \Delta-\widehat{\Delta}(t)$, where the
dimensions of $\Delta$ (and in turn, $\widehat{\Delta}(t)$ and
$e_\Delta(t)$ will be clear from the context.


\subsection{The case of bounded state-dependent uncertainties}
\label{sec:case2} We consider the case where the unknown $\|\Delta
b (\xi(t))\|$ is upper bounded by a function of the state
$\xi(t)$, i.e.
\begin{align}
\label{eq:upb} \|\Delta b(\xi(t))\|~\leq~\|\Delta\|\|L(\xi)\|,
\end{align}
where $\Delta\in\mathbb{R}^{m\times m}$ is constant, and $L(\xi)$
is a known bounded state function. Assume, for now, that we can
obtain the estimate of $\Delta(i,j)$, which may be time-varying
and is denoted by $\widehat{\Delta}(i,j)$, for $i,j =
1,2,\dots,m$. Let $\widehat{\Delta}(t)$
 be the matrix with the element $\widehat{\Delta}(i,j)$.
 We use the following robust controller
\begin{align}
\label{eq:case2}
u_{r}=&~-A^{-1}(\xi)\tilde B ^TP z\|L(\xi)\|^2 \notag\\
&~-A^{-1}(\xi) \|\widehat{\Delta}(t)\|\|L(\xi)\|{\rm{sign} }  (\tilde B ^TP z).
\end{align}
The closed-loop error dynamics can be written in the form of
\begin{align}
\label{eq:cl_case2}
\dot{z}~=~f(t,z,e_\Delta),
\end{align}
where $e_\Delta(t)$ is the system input and $z(t)$ is the system state.

\begin{theorem}
\label{thm:case2}
Consider the system (\ref{eq:model}), under Assumptions 1-5 and
the assumption that $\Delta b(\xi(t))$ satisfies (\ref{eq:upb}),
with the feedback controller (\ref{eq:full}), where $u_n$ is given
by (\ref{eq:norm}) and $u_r$ is given by (\ref{eq:case2}). Then,
the closed-loop system (\ref{eq:cl_case2}) is ISS from the
estimation errors input $e_\Delta(t)\in\mathbb{R}^{m\times m}$ to
the tracking errors state $z(t)\in\mathbb{R}^n$.
\end{theorem}

\begin{proof}
By substitution (\ref{eq:case2}) into (\ref{eq:delta}), we obtain
$\delta=-\tilde B ^TP z\|L(\xi)\|^2 -
\|\widehat{\Delta}(t)\|\|L(\xi)\|{\rm{sign} }  (\tilde B ^TP z) +
\Delta b(\xi(t))$. We consider $V(z) =z^TPz$ as a Lyapunov
function for the error dynamics (\ref{eq:dyn_z}), where $P>0$ is a
solution of (\ref{eq:lyap}). We can derive that
\begin{align}
\label{eq:storage} {\lambda}_{\min}(P)\|z\|^2\leq V(z)\leq
{\lambda}_{\max}(P)\|z\|^2,
\end{align}
where ${\lambda}_{\min}(P)>0$,  ${\lambda}_{\max}(P)>0$ denote
respectively the minimum and the maximum eigenvalues of the matrix
$P$. Then, from (\ref{eq:general}), we obtain
\begin{align*}
\dot{V}
~=&~ -\|z\|^2 + 2 z^T P \tilde B \Delta b(\xi(t))\\
&~-2\|z^T P \tilde B\|^2\|L(\xi)\|^2-2\|z^T P \tilde B\|_1\|\widehat{\Delta }(t)\|\|L(\xi)\|.
\end{align*}
Since $\|z^T P \tilde B\|~\leq~\|z^T P \tilde B\|_1$, we have
\begin{align*}
\dot{V}
~\leq&~ -\|z\|^2 + 2 z^T P \tilde B \Delta b(\xi(t))\\
&~-2\|z^T P \tilde B\|^2\|L(\xi)\|^2-2\|z^T P \tilde B\|\|\widehat{\Delta }(t)\|\|L(\xi)\|.
\end{align*}
Then based on the assumption (\ref{eq:upb}) and the fact that $z^T
P \tilde B \Delta b(\xi(t)) \leq \|z^T P \tilde B\|\| \Delta
b(\xi(t))\|$, we obtain
\begin{align*}
\dot{V}
~\leq&~ -\|z\|^2 + 2 \|z^T P \tilde B\| \|\Delta \|\|L(\xi)\|\\
&~-2\|z^T P \tilde B\|^2\|L(\xi)\|^2-2\|z^T P \tilde B\|\|\widehat{\Delta }(t)\|\|L(\xi)\|\\
~=&~-\|z\|^2-2\|z^T P \tilde B\|^2\|L(\xi)\|^2\\
&~+ 2 \|z^T P \tilde B\| \|L(\xi)\|(\|\Delta \|-\|\widehat{\Delta }(t)\|).
\end{align*}
Because $\|\Delta\|-\|\widehat{\Delta }(t)\|\leq\|e_\Delta\|$, we obtain
\begin{align*}
\dot{V}
~\leq&~-\|z\|^2-2\|z^T P \tilde B\|^2\|L(\xi)\|^2\\
&~+ 2 \|z^T P \tilde B\| \|L(\xi)\|\|e_\Delta\|.
\end{align*}
Further, we can obtain
\begin{align*}
\dot{V}&~\leq~-\|z\|^2 -2(\|z^T P \tilde B\|\|L(\xi)\|-\frac{1}{2}\|e_\Delta \|)^2 + \frac{1}{2}\|e_\Delta\|^2\\
&~\leq~-\|z\|^2+ \frac{1}{2}\|e_\Delta\|^2.
\end{align*}
Thus, we have the following relation
\begin{align*}
\dot{V}\leq-\frac{1}{2}\|z\|^2,\quad \forall \|z\|\geq \|e_\Delta\|>0,
\end{align*}
Then from (\ref{eq:storage}), we obtain that system (\ref{eq:cl_case2}) is ISS from input $e_\Delta$ to state $z$ as guaranteed by Theorem \ref{thm:iss}.
\end{proof}


\subsection{Multi-parametric ES-based uncertainties estimation}
\label{sec:ES}
Let us define now the following cost function
\begin{align}
\label{eq:cost_gen}
J(\widehat{\Delta}) ~=~ F(z(\widehat{\Delta}))
\end{align}
where $F : \mathbb{R}^{n} \rightarrow \mathbb{R}$, $F(0) = 0$, $F(z)>0$ for $z\neq0$. We need the following assumptions on $J$.

\begin{assumption}
\label{asp:cost}
The cost function $J$ has a local minimum at $\widehat{\Delta}^* ~=~\Delta$.
\end{assumption}

\begin{assumption}
\label{asp:local}
The initial error $e_\Delta(t_0)$ is sufficiently small, i.e. the original parameter estimate vector or matrix $\widehat{\Delta}$ are close enough to the actual parameter vector or matrix $\Delta$.
\end{assumption}

\begin{assumption}
\label{asp:Lips}
The cost function $J$ is analytic and its variation with respect to the uncertain parameters is bounded in the neighborhood of $\widehat{\Delta}^*$, i.e. $\|\frac{\partial J}{\partial \widehat{\Delta}}(\tilde \Delta)\|\leq \xi_2$, $\xi_2>0$,
$\tilde{\Delta}\in\mathcal{V}(\widehat{\Delta}^*)$, where $\mathcal{V}(\widehat{\Delta}^*)$ denotes a compact neighborhood of $\widehat{\Delta}^*$.
\end{assumption}

\begin{remark}
Assumption \ref{asp:cost} simply states that the cost function $J$ has at least a local minimum at the true values of the uncertain parameters.
\end{remark}

\begin{remark}
Assumption \ref{asp:local} indicates that out results will be of local nature, i.e. our analysis holds in a small neighborhood of the actual values of the uncertain parameters.
\end{remark}

We can now state the following result.
\begin{lemma}
\label{lem:case1} Consider the system (\ref{eq:model}) with the
cost function (\ref{eq:cost_gen}), under Assumptions 1-8 and the
assumption that $\Delta b(\xi(t)) = [\Delta_1,\dots,\Delta_m]^T$,
with the feedback controller (\ref{eq:full}), where $u_n$ is given
by (\ref{eq:norm}) and $u_r$ is given by (\ref{eq:case2}), and
$\widehat{\Delta}(t)$ is estimated through the MES algorithm
\begin{align}
\label{eq:mes}
\dot{x}_i~&=~a_i \sin(\omega_i t+ \frac{\pi}{2})J(\widehat{\Delta})\notag\\
\widehat{\Delta}_i(t)~&=~x_i+a_i\sin(\omega_i t - \frac{\pi}{2}), \quad i\in\{1,2,\dots,m\}
\end{align}
with $\omega_i\neq \omega_j$, $\omega_i + \omega_j\neq\omega_k$,
$i,j,k\in\{1,2,\dots,m\}$, and $\omega_i>\omega^*$, $\forall
~i\in\{1,2,\dots,m\}$, with $\omega^*$ large enough, ensures that
the norm of the error vector $z(t)$ admits the following bound
\begin{align*}
\|z(t)\|&\leq\beta(\|z(0)\|,t)
+ \gamma(\tilde{\beta}(\|e_\Delta(0)\|,t)+\|e_\Delta\|_{\max})
\end{align*}
where $\|e_\Delta\|_{\max} = \frac{\xi_1}{\omega_0} + \sqrt{\sum_{i = 1}^m a_i^2}$, $\xi_1>0$, $e_\Delta(0)\in\mathcal{D}_e$, $\omega_0 = \max_{i\in\{1,2,\dots,m\}} \omega_i$, $\beta\in\mathcal{KL}$, $\tilde{\beta}\in\mathcal{KL}$ and $\gamma\in\mathcal{K}$.
\end{lemma}

\begin{proof}
Based on Theorem \ref{thm:case2}, we know that the tracking error
dynamics (\ref{eq:cl_case2}) is ISS from the input $e_\Delta(t)$
to the state $z(t)$. Thus, by Definition \ref{def:ISS}, there
exist a class $\mathcal{KL}$ function $\beta$ and a class
$\mathcal{K}$ function $\gamma$ such that for any initial state
$z(0)$, any bounded input $e_\Delta(t)$ and any $t\geq0$,
\begin{align}
\label{eq:bound}
\|z(t)\|\leq \beta(\|z(0)\|,t) + \gamma(\sup_{0\leq\tau\leq t} \|e_\Delta(\tau)\|).
\end{align}
Now, we need to evaluate the bound on the estimation vector
$\widehat{\Delta}(t)$, to do so we use the results presented in
\cite{Rote00}. First, based on Assumption \ref{asp:Lips}, the cost
function is locally Lipschitz, i.e. there exists $\eta_1>0$ such
that $|J(\Delta_1)-J(\Delta_2)|\leq\eta_1\|\Delta_1 -\Delta_2\|$,
for all $\Delta_1, \Delta_2 \in\mathcal{V}(\widehat{\Delta}^*)$.
Furthermore, since $J$ is analytic, it can be approximated locally
in $\mathcal{V}(\widehat{\Delta}^*)$ by a quadratic function, e.g.
Taylor series up to the second order. Based on this and on
Assumptions \ref{asp:cost} and \ref{asp:local}, we can obtain the
following bound (\cite[p. 436-437]{Rote00},\cite{BA13})
\begin{align*}
\|e_\Delta(t)\| - \|d(t)\| \leq \|e_\Delta(t)- d(t)\|\leq\tilde{\beta}(\|e_\Delta(0),t\|) + \frac{\xi_1}{\omega_0},
\end{align*}
where $\tilde{\beta}\in\mathcal{KL}$, $\xi_1>0$, $t\geq0$,
$\omega_0 = \max_{i\in\{1,2,\dots,m\}} \omega_i$, and $d(t) =
[a_1\sin(\omega_1 t + \frac{\pi}{2}),\dots,a_m\sin(\omega_m
t+\frac{\pi}{2})]^T$. We can further obtain that
\begin{align*}
\|e_\Delta(t)\|&\leq~ \tilde{\beta}(\|e_\Delta(0),t\|) + \frac{\xi_1}{\omega_0} + \|d(t)\|\\
&\leq~\tilde{\beta}(\|e_\Delta(0),t\|) + \frac{\xi_1}{\omega_0} + \sqrt{\sum_{i = 1}^m a_i^2}.
\end{align*}
Together with (\ref{eq:bound}) yields the desired result.
\end{proof}

\section{The case of nonlinear system model with time dependent
uncertainties} We consider here affine uncertain nonlinear systems
of the form:
\begin{equation}
\begin{array}{l}
\displaystyle \dot{x} ~=~ f(x) + \Delta f(t,x) + g(x)u\\
\displaystyle y~=~ h(x),
\end{array}
\label{eq:model}
\end{equation}
where $x\in\mathbb{R}^n$, $u\in \mathbb{R}^p$, $y\in\mathbb{R}^m$
($p~\geq~ m$), represent respectively the state, the input and the
controlled output vectors, $\Delta f(t,x)$ is a vector field
representing additive model uncertainties. The vector fields $f$,
$\Delta f$, columns of $g$ and function $h$ satisfy the following
assumptions.

\begin{assumption}
The function $f:\mathbb{R}^n \rightarrow \mathbb{R}^n$ and the
columns of $g: \mathbb{R}^n \rightarrow \mathbb{R}^p$ are
$\mathbb{C}^{\infty}$ vector fields on a bounded set $X$ of
$\mathbb{R}^n$ and $h : \mathbb{R}^n \rightarrow \mathbb{R}^m$ is
a $\mathbb{C}^{\infty}$ vector on $X$. The vector field $\Delta f
$ is $\mathbb{C}^1$ on $X$.
\end{assumption}

\begin{assumption}
System (\ref{eq:model}) has a well-defined (vector) relative
degree $\{r_1,~r_2,~\cdots,~r_m\}$ at each point $x^0\in X$, and
the system is linearizable, i.e. $\sum_{i=1}^m r_i=n$.
\end{assumption}


\begin{assumption}
The desired output trajectories $y_{id}$ ($1\leq i \leq m$) are
smooth functions of time, relating desired initial points
$y_{id}(0)$ at $t=0$ to desired final points $y_{id}(t_f)$ at
$t=t_f$.
\end{assumption}

\subsection{Control objectives}
Our objective is to design a state feedback controller $u(x)$ so
that  for the uncertain nonlinear model (\ref{eq:model}) the
tracking error is uniformly bounded. We stress here that the goal
of parameters tuning is not for stabilization but for performance
optimization. To achieve the control objective, we proceed as
follows. First, we design a robust controller which can guarantee
the input-to-state stability (ISS) of the tracking error dynamics
w.r.t the estimation errors input. Then, we combine this
controller with a model-free extremum-seeking algorithm to
iteratively tune the uncertain parameters, to optimize online a
desired performance cost function.

\subsection{Nominal Controller Design} \label{sec:norm}
Under Assumption 2 and nominal conditions, i.e. when $\Delta
f(t,x) = 0$, system (\ref{eq:model}) can be written as, e.g.
\cite{I89}
\begin{align}
\label{eq:linear} y^{(r)}(t) ~=~ b(\xi(t)) + A (\xi(t)) u(t),
\end{align}
where
\begin{equation}
\begin{array}{l}
\displaystyle y^{(r)}(t) ~=~ [y^{(r_1)}_{1}(t), ~y^{(r_2)}_{2}(t),~\cdots,~ y^{(r_m)}_{m}(t)]^T\\
\displaystyle \xi(t) ~=~ [\xi^1(t),~\cdots,~\xi^m(t)]^T\\
\displaystyle \xi ^i(t) ~=~[y_i(t),~\cdots,~y_i^{(r_i-1)}(t)],
\quad 1\leq i\leq m
\end{array}
\label{eq:linear_c}
\end{equation}
The functions $b(\xi)$, $A(\xi)$ can be written as functions of
$f$, $g$ and $h$, and $A(\xi)$ is non-singular in $\tilde{X}$,
where $\tilde X$ is the image of the set of $X$ by the
diffeomorphism $x\rightarrow\xi$ between the states of system
(\ref{eq:model}) and the linearized model (\ref{eq:linear}).

At this point, we introduce one more assumption on system
(\ref{eq:model}).

\begin{assumption}
The additive uncertainties $\Delta f(t,x)$ in (\ref{eq:model})
appear as additive uncertainties in the input-output linearized
model (\ref{eq:linear})-(\ref{eq:linear_c}) as follows (see also
\cite{Beno09})
\begin{align}
\label{eq:model_un} y^{(r)}(t) ~=~ b(\xi(t)) + A (\xi(t)) u(t) +
\Delta b(t,\xi(t)),
\end{align}
where $\Delta b(t,\xi(t))$ is $\mathbb{C}^1$ on $\tilde X$.
\end{assumption}

If we consider the nominal model (\ref{eq:linear}) first, then we
can define a virtual input vector $v(t)$ as
\begin{align}
\label{eq:input} v(t) ~=~ b(\xi(t)) + A (\xi(t)) u(t).
\end{align}
Combining (\ref{eq:linear}) and (\ref{eq:input}), we can obtain
the following input-output mapping
\begin{align}
\label{eq:virtual} y^{(r)}(t) ~=~v(t).
\end{align}
Based on the linear system (\ref{eq:virtual}), it is
straightforward to apply a stabilizing controller for the nominal
system (\ref{eq:linear}) as
\begin{align}
\label{eq:norm} u_{n} = A^{-1}(\xi) \left[v_{s}(t,\xi) -
b(\xi)\right],
\end{align}
where $v_s$ is a $m\times1$ vector and the $i$-th ($1\leq i \leq
m$) element $v_{si}$ is given by
\begin{align}
\label{eq:feedback} v_{si} = y_{id}^{(r_i)} -
K_{r_i}^{i}(y_i^{(r_i-1)}-y_{id}^{(r_i-1)})-\cdots-K_{1}^{i}(y_i-y_{id}).
\end{align}
Denote the tracking error as $e_i(t) ~\triangleq~ y_i(t)-
y_{id}(t)$, we obtain  the following tracking error dynamics
\begin{align}
\label{eq:error} e_i^{(r_i)}(t) +  K^i_{r_i}e^{(r_i-1)}(t) +
\cdots + K^i_1e_i(t) = 0,
\end{align}
where $i \in \{1,~2,~\cdots,~m\}$. By selecting the gains $K^i_j$
where $i \in \{1,~2,~\cdots,~m\}$ and $j\in\{1,~2,~\cdots,~r_i\}$,
we can obtain global asymptotic stability of the tracking errors
$e_i(t)$. To formalize this condition, we make the following
assumption.

\begin{assumption}
There exists a non-empty set $\mathcal{A}$ where
$K^i_j\in\mathcal{A} $ such that the polynomials in
(\ref{eq:error}) are Hurwitz, where $i \in \{1,~2,~\cdots,~m\}$
and $j\in\{1,~2,~\cdots,~r_i\}$.
\end{assumption}


To this end, we define $z=[z^1,~z^2,~\cdots,~z^m]^T$, where
$z^i=[e_i,~\dot{e_i},~\cdots,~e_i^{(r_i-1)}]$ and
$i\in\{1,~2,\cdots,~m\}$. Then, from (\ref{eq:error}), we can
obtain
\begin{align*}
\dot{z}~=~ \tilde{A} z,
\end{align*}
where $\tilde{A}\in \mathbb{R}^{n\times n}$ is a diagonal block
matrix given by
\begin{align}
\label{eq:tildA}
\tilde{A}~=~{\rm{diag}}\{\tilde A_1,~\tilde A_2,~\cdots,~\tilde
A_m\},
\end{align}
and $\tilde A_i$ ($1\leq i\leq m$) is a $r_i\times r_i$ matrix
given by
\[ \tilde{A}_i~=~\left[ \begin{array}{cccccc}
0& 1&  \\
0& &1&\\
0&&&\ddots\\
\vdots&&&&1\\
-K^i_1\quad&-K^i_2\quad&\cdots\quad&\cdots\quad&-K^i_{r_i}
\end{array} \right].\]
As discussed above, the gains $K_j^i$ can be chosen so that the
matrix $\tilde{A}$ is Hurwitz. Thus, there exists a positive
definite matrix $P~>~ 0$ such that (see e.g. \cite{Khalil02})
\begin{align}
\label{eq:lyap} \tilde A^T P + P \tilde A ~=~-I.
\end{align} 


\subsection{Robust Controller Design} \label{sec:robust}
We now consider the uncertain model (\ref{eq:model}), i.e. when
$\Delta f(t,x)\neq0$. The corresponding linearized model is given
by (\ref{eq:model_un}) where $\Delta b (t,\xi(t))\neq0$. The
global asymptotic stability of the error dynamics (\ref{eq:error})
cannot be guaranteed anymore due to the additive uncertainty
$\Delta b(t,\xi(t))$. We use  Lyapunov reconstruction techniques
to design a new controller so that the tracking error is
guaranteed to be bounded. The new controller for the uncertain
model (\ref{eq:model_un}) is defined as
\begin{align}
\label{eq:full} u_{f}~=~u_{n} + u_{r},
\end{align}
where the nominal controller $u_{n}$ is given by (\ref{eq:norm})
and the robust controller $u_r$ will be given later on based on
particular forms of the uncertainty $\Delta b(t,\xi(t))$. By using
the controller (\ref{eq:full}), from (\ref{eq:model_un}) we obtain
\begin{align}
\label{eq:io_new}
y^{(r)}(t) &~=~ b(\xi(t)) + A (\xi(t)) u_{f} + \Delta b(t,\xi(t))\notag\\
&~=~ b(\xi(t)) + A (\xi(t)) u_{n} + A (\xi(t)) u_{r}+ \Delta b(t,\xi(t))\notag\\
& ~=~v_{s}(t,\xi)+ A (\xi(t)) u_{r}+ \Delta b(t,\xi(t)),
\end{align}
Further, the dynamics for $z$ is given by
\begin{align}
\label{eq:dyn_z} \dot{z}~=~ \tilde{A} z + \tilde{B} \delta,
\end{align}
where $\tilde{A}$ is defined in (\ref{eq:tildA}), $\delta$ is a
$m\times 1$ vector given by
\begin{align}
\label{eq:delta} \delta~=~A (\xi(t)) u_{r}+ \Delta b(t,\xi(t)),
\end{align}
and the matrix $\tilde B\in\mathbb{R}^{n\times m}$ is given by
\begin{align}
\label{eq:tildB} \tilde{B}~=~\left[ \begin{array}{c}
\tilde{B}_1 \\
\tilde{B}_2\\
\vdots\\
\tilde{B}_m
\end{array} \right],
\end{align}
with $\tilde B_i$ ($1\leq i\leq m$) given by a $r_i\times m$
matrix such that
\[ \tilde{B}_i(l,q) ~=~ \left\{ \begin{array}{lll}
         1 & \quad\mbox{if $l = r_i$ and $q = i$}\\
         0  & \quad\mbox{otherwise}\end{array} \right. \]
%
%
If we apply $V(z)=z^TPz$ as a Lyapunov function for the dynamics
(\ref{eq:dyn_z}), where $P$ is the solution of the Lyapunov
equation (\ref{eq:lyap}), then we obtain
\begin{align}
\label{eq:general}
\dot{V}(t)&~=~\frac{\partial V}{\partial z}\dot{z}\notag\\
&~=~z^T(\tilde A^T P + P \tilde A)z + 2 z^T P \tilde B \delta\notag\\
&~=~-\|z\|^2 + 2 z^T P \tilde B \delta,
\end{align}
%
where $\delta$ given by (\ref{eq:delta}) depends on the robust
controller $u_r$.

Next, we will design the controller $u_{r}$ based on the
particular forms of the uncertainties that appear in
(\ref{eq:model_un}), i.e. $\Delta b(t,\xi(t))$. For notational
convenience, the unknown parameter vector/matrix (which may be
time-varying) is denoted by $\Delta(t)$ and the estimate for the
unknowns is denoted by $\widehat{\Delta}(t)$. Further, the
estimation error vector/matrix is given by $e_\Delta(t) =
\Delta(t)-\widehat{\Delta}(t)$. The dimensions of $\Delta$ (and in
turn, $\widehat{\Delta}$ and $e_\Delta$) will be clear from the
context.

\subsection{Case 1: State-Independent Uncertainties}
\label{sec:case1} We consider the case when $\Delta b(t,\xi(t))$
is simply $\Delta(t) $, where $\Delta(t)  =
[\Delta_1(t),\dots,\Delta_m(t)]^T$. Assume that we can obtain the
estimate (e.g. by ES) of the unknown parameters $\Delta_i(t)$,
which may be time-varying and is denoted by $\widehat{\Delta }_i
(t)$, for $i = 1,2,\dots,m$.  Let $\widehat{\Delta }(t) =
[\widehat{\Delta } _1(t),\dots,\widehat{\Delta } _m(t)]^T$.
We use the following robust controller
\begin{align}
\label{eq:case1}
u_{r} = - A^{-1}(\xi) (\tilde B^TPz+ \widehat{\Delta}(t)).
\end{align}
The closed-loop error dynamics can be written as
\begin{align}
\label{eq:cl_case1} \dot{z}~=~f(z,e_\Delta),
\end{align}
where $e_\Delta(t)$ is the input to the system, $z(t)$ represents
the system state and $f$ is given by
\begin{align*}
f(z,e_\Delta)~=~(\tilde{A} - \tilde{B} \tilde B^TP)z + \tilde B
e_\Delta.
\end{align*}
\begin{theorem}
\label{thm:case1} Consider the system (\ref{eq:model}), under
Assumptions 1-5 and the assumption that $\Delta b(t,\xi(t)) =
[\Delta_1(t),\dots,\Delta_m(t)]^T$, with the feedback controller
(\ref{eq:full}), where $u_n$ is given by (\ref{eq:norm}) and $u_r$
is given by (\ref{eq:case1}). Then, the closed-loop system
(\ref{eq:cl_case1}) is ISS from the estimation errors input
$e_\Delta(t)\in\mathbb{R}^m$
 to the tracking errors state $z(t)\in\mathbb{R}^n$.
%
\end{theorem}

\subsection{Case2: State-Dependent Uncertainties}
\label{sec:case2} We consider the second case when $\|\Delta b
(t,\xi(t))\|$ is upper bounded by a function of the state
$\xi(t)$, i.e.
\begin{align}
\label{eq:upb} \|\Delta
b(t,\xi(t))\|~\leq~\|\Delta(t)\|\|L(\xi)\|,
\end{align}
where $\Delta(t)\in\mathbb{R}^{m\times m}$ and $L(\xi)$ is a known
bounded function. Assume that we can obtain the estimate (e.g. by
ES) for $\Delta(i,j)$, which may be time-varying and is denoted by
$\widehat{\Delta}(i,j)$, for $i,j = 1,2,\dots,m$. Let
$\widehat{\Delta}(t)$
 be the matrix with the element $\widehat{\Delta}(i,j)$.
 We use the following robust controller
\begin{align}
\label{eq:case2}
u_{r}=&~-A^{-1}(\xi)\tilde B ^TP z\|L(\xi)\|^2 \notag\\
&~-A^{-1}(\xi) \|\widehat{\Delta}(t)\|\|L(\xi)\|{\rm{sign} }
(\tilde B ^TP z).
\end{align}
Similar to the previous case, the closed-loop error dynamics can
be written in the form of
\begin{align}
\label{eq:cl_case2} \dot{z}~=~f(t,z,e_\Delta),
\end{align}
where $e_\Delta(t)$ is the system input and $z(t)$ is the system
state.

\begin{theorem}
\label{thm:case2}
Consider the system (\ref{eq:model}), under Assumptions 1-5 and
the assumption that $\Delta b(t,\xi(t))$ satisfies (\ref{eq:upb}),
with the feedback controller (\ref{eq:full}), where $u_n$ is given
by (\ref{eq:norm}) and $u_r$ is given by (\ref{eq:case2}). Then,
the closed-loop system (\ref{eq:cl_case2}) is ISS from the
estimation errors input $e_\Delta(t)\in\mathbb{R}^{m\times m}$ to
the tracking errors state $z(t)\in\mathbb{R}^n$.
\end{theorem}

\subsection{Case 3: Sum of a State-dependent Term and a Time-dependent Term}
\label{sec:case3} We consider the third case when $\Delta b
(t,\xi(t))$ is composed of a  state-dependent term and a
time-dependent term, i.e.
\begin{align}
\label{eq:affine} \Delta b (t,\xi(t))~=~\Delta(t) (Q(\xi) +
\eta(t)),
\end{align}
where $\Delta(t)\in\mathbb{R}^{m\times m}$, $Q(\xi)$ is a known
bounded function, the vector $\eta(t)$ is unknown but the upper
bound for $\|\eta(t)\|$ is known to be $C_1$, i.e.
$\|\eta(t)\|\leq C_1$. Assume that we can obtain the estimate
(e.g. by ES) for $\Delta(i,j)$, which may be time-varying and is
denoted by $\widehat{\Delta}(i,j)$, for $i,j = 1,2,\dots,m$. Let
$\widehat{\Delta}(t)$
 be the matrix with the element $\widehat{\Delta}(i,j)$ that locates at the $i$-th row and $j$-th column.
 We use the following robust controller
\begin{align}
\label{eq:case3}
u_{r}~=~&-A^{-1}(\xi)[\tilde B^TPz\|Q(\xi)\|^2 +\widehat{\Delta}(t)\times Q(\xi)\notag\\
 &+\|\widehat \Delta(t)\| C_1{\rm{sign} }  (\tilde B^TPz) +\tilde B^TPzC_1^2].
\end{align}
Similar to the previous two cases, the closed-loop error dynamics
can be written in the following form
\begin{align}
\label{eq:cl_case3} \dot{z}~=~F(t,z,e_\Delta),
\end{align}
where $e_\Delta(t)$ is the system input and $z(t)$ is the system
state.

\begin{theorem}
\label{thm:case3} Consider the system (\ref{eq:model}), under
Assumptions 1-5 and the assumption that $\Delta b(t,\xi(t))$
satisfies (\ref{eq:affine}), with the feedback controller
(\ref{eq:full}), where $u_n$ is given by (\ref{eq:norm}) and $u_r$
is given by (\ref{eq:case3}). Then, the closed-loop system
(\ref{eq:cl_case3}) is ISS from the estimation errors input
$e_\Delta(t)\in\mathbb{R}^{m\times m}$ to the tracking errors
state $z(t)\in\mathbb{R}^n$.
\end{theorem}

\subsection{Multi-parametric ES-based Adaptation} \label{sec:ES}

Let us now define the following cost function
\begin{align}
\label{eq:cost_gen}
J(\widehat{\Delta},t)~=~F(z(\widehat{\Delta}),t)
\end{align}
where $F:\mathbb{R}^{n}\times \mathbb{R}^+\rightarrow
\mathbb{R}^+$, $F(0,t) = 0$, $F(z,t)>0$ for $z\neq 0$. We need the
following assumptions on $J$.

\begin{assumption}
\label{asp:cost} The cost function $J$ has a local minimum at
$\widehat{\Delta}^*(t) = \Delta(t)$.
\end{assumption}

\begin{assumption}
\label{asp:Lips} $|\frac{\partial J (\widehat{\Delta},t)}{\partial
t}|<\rho_J$, for any $t\in\mathbb{R}^+$ and any
$\widehat{\Delta}\in\mathbb{R}^p$.
\end{assumption}

\begin{remark}
Assumption \ref{asp:cost} simply states that the cost function $J$
has at least a local minimum at the true values of the uncertain
parameters.
\end{remark}

We can now present the following result for Case 1, i.e. the case
that is studied in Section \ref{sec:case1}.

\begin{lemma}
Consider the system (\ref{eq:dyn_z}) with the cost function
(\ref{eq:cost_gen}), under Assumptions
\ref{asp:cost}-\ref{asp:Lips} and the assumption that $\Delta(t)
= [\Delta_1(t),\dots,\Delta_m(t)]^T$, with the feedback controller
(\ref{eq:full}), where $u_n$ is given by (\ref{eq:norm}) and $u_r$
is given by (\ref{eq:case1}), and $\widehat{\Delta}(t)$ is
estimated through the MES algorithm
\begin{align}
\label{eq:esa_gen} \dot{\widehat{\Delta}}_i~=~a_i\sqrt{\omega_i}
\cos(\omega_i t)- k_i \sqrt{\omega_i}\sin(\omega_i t)
J(\widehat{\Delta},t),
\end{align}
where $i\in\{1,2,\dots, p\}$, $a_i>0$, $k_i>0$,
$\omega_i\neq\omega_j$, and $\omega_i>\omega^*$, with $\omega^*$
large enough, ensures that the norm of the tracking error admits
the following bound
\begin{align*}
\|z(t)\|\leq \beta(\|z(t_0)\|,t) + \gamma(\sup_{0\leq\tau\leq t}
\|e_\Delta(\tau)\|)
\end{align*}
where $\beta\in\mathcal{KL}$, $\gamma\in\mathcal{K}$ and
$\|e_\Delta\|$ satisfies:
\begin{enumerate}
\item $(\frac{1}{\omega},d)$-Uniform Stability: For every
$c_2\in(d,\infty)$, there exists $c_1\in(0,\infty)$ and
$\hat{\omega}>0$ such that for all $t_0\in\mathbb{R}$ and for all
$e_\Delta(0)\in\mathbb{R}^m$ with $\|e_\Delta(0)\|<c_1$ and for
all $\omega>\hat{\omega}$,
\begin{align*}
\|e_\Delta(t,e_\Delta(0))\|<c_2, \quad \forall t\in[t_0,\infty)
\end{align*}
\item $(\frac{1}{\omega},d)$-Uniform Ultimate Boundedness: For
every $c_1\in(0,\infty)$, there exists $c_2\in(d,\infty)$ and
$\hat{\omega}>0$ such that for all $t_0\in \mathbb{R}$ and for all
$e_\Delta(0)\in\mathbb{R}^m$ with $\|e_\Delta(0)\|<c_1$ and for
all $\omega>\hat{\omega}$,
\begin{align*}
\|e_\Delta(t,e_\Delta(0))\|<c_2, \quad \forall t\in[t_0,\infty)
\end{align*}
\item $(\frac{1}{\omega},d)$-Global Uniform Attractivity: For all
$c_1,~c_2 \in(d,\infty)$ there exists $T\in[0,\infty)$ and
$\hat{\omega}>0$ such that for all $t_0\in\mathbb{R}$ and for all
$e_\Delta(0)\in\mathbb{R}^m$ with $\|e_\Delta(0)\|<c_1$ and for
all $\omega>\hat{\omega}$,
\begin{align*}
\|e_\Delta(t,e_\Delta(0))\|<c_2, \quad \forall t\in[t_0+T,\infty)
\end{align*}
\end{enumerate}
\end{lemma}

Similar bounds can be derived for the two remaining cases but are
omitted here because of space constraints.

\section{Mechatronic Example}
\label{sec:sim}

We consider here a two-link robot manipulator example. The
dynamics for the manipulator in the nominal case, is given by (see
e.g. \cite{Spong92})
\begin{align}
\label{eq:robot}
H(q)\ddot{q} + C(q,\dot{q})\dot{q}+ G(q)~=~\tau,
\end{align}
where $q \triangleq [q_1,q_2]^T$ denotes the two joint angles and $\tau \triangleq [\tau_1,\tau_2]^T$ denotes the two joint torques. The matrix
$H$ is assumed to be non-singular and is given by
\[ H~\triangleq~\left[ \begin{array}{ccc}
H_{11} & H_{12} \\
H_{21}& H_{22}\\
\end{array} \right]\]
where
\begin{equation}
\label{eq:robot_c}
\begin{array}{l}
\displaystyle H_{11} ~=~m_1\ell_{c_1}^2+I_1+m_2[\ell_1^2+\ell_{c_2}^2+ 2\ell_1\ell_{c_2}\cos(q_2)]+I_2 \\
\displaystyle H_{12} ~=~ m_2\ell_1\ell_{c_2}\cos(q_2)+m_2\ell_{c_2}^2+I_2\\
\displaystyle H_{21} ~=~H_{12}\\
\displaystyle H_{22} ~=~ m_2\ell_{c_2}^2+I_2
\end{array}
\end{equation}
The matrix $C(q,\dot{q})$ is given by
\[ C(q,\dot{q})~\triangleq~\left[ \begin{array}{ccc}
-h\dot{q_2} \quad& -h\dot{q_1}-h\dot{q_2} \\
h\dot{q_1}\quad&0\end{array} \right],\]
where $h~=~m_2\ell_1\ell_{c_2}\sin(q_2)$. The vector $G ~=~ [G_1,G_2]^T$ is given by
\begin{equation}
\begin{array}{l}
\displaystyle G_1 =~m_1\ell_{c_1}g\cos(q_1)+m_2 g[\ell_2\cos(q_1+q_2)+\ell_1\cos(q_1)]\\
\displaystyle G_2 =~ m_2\ell_{c_2}g\cos(q_1+q_2)
\end{array}
\label{eq:robot_g}
\end{equation}
In our simulations, we assume that the parameters take values according to \cite{Spong92} summarized in Table \ref{tab:robot}.
%
\begin{table}
\caption{System Parameters for the manipulator example.}
\label{tab:robot}
\centering
\begin{tabular}{|c|c|}
  \hline
  Parameter & Value \\\hline
  $I_2$ & $\frac{5}{12}~ [kg\cdot m^2]$ \\\hline
  $m_1$ & $10~ [kg]$ \\\hline
  $m_2$ & $5~ [kg]$\\\hline
  $\ell_1$ & $1~ [m]$ \\\hline
  $\ell_2$ & $1 ~[m]$ \\\hline
  $\ell_{c_1}$ & $0.5 ~[m]$ \\\hline
  $\ell_{c_2}$  & $0.5~ [m]$\\\hline
  $I_1$ & $\frac{10}{12}~[kg\cdot m^2]$ \\\hline
  $g$ & $9.8~ [m/s^2]$ \\
  \hline
\end{tabular}
\end{table}
The system dynamics (\ref{eq:robot}) can be rewritten as
\begin{align}
\label{eq:norm_robot}
\ddot{q}~=~H^{-1}(q) \tau -  H^{-1}(q) \left[C(q,\dot{q})\dot{q}+ G(q)\right],
\end{align}
Thus, the nominal controller is given by
\begin{align}
\tau_{n}~=&~\left[C(q,\dot{q})\dot{q}+ G(q)\right] \notag\\
&+~ H(q)\left[\ddot{q_{d}} - K_2(\dot{q}-\dot{q_{d}})-K_1(q-q_{d})\right],
\end{align}
where $q_d =[q_{1d},q_{2d}]^T$, denotes the desired trajectory and the feedback gains $K_1>0$, $K_2>0$, are chosen such that the tracking error will go to zero asymptotically. For simplicity, we use the feedback gains $K^i_j = 1$ in (\ref{eq:feedback}) for $i = 1,2$ and $j = 1,2$ in our simulations.
The reference trajectory is given by the following function from the initial time $t_0 = 0$ to the final time $t_f$, where
\begin{align*}
q_{id}(t)~=~\frac{1}{1+\exp{(-t)}}, \quad  i~=~1,2
\end{align*}
\subsection{State dependent uncertainties}
Now we introduce an uncertain term to the nonlinear model
(\ref{eq:robot}). In particular, we assume that there exist
additive uncertainties in the model (\ref{eq:norm_robot}), i.e.
\begin{align}
\label{eq:uncertain_robot} \ddot{q}=~H^{-1}(q) \tau -  H^{-1}(q)
\left[C(q,\dot{q})\dot{q}+ G(q)\right] + \Delta b(q).
\end{align}
We assume additive uncertainties on the gravity vector, such that
\begin{align}
\label{eq:robot_case2} \Delta b(q) =  \Delta \times G(q),
\end{align}
so that we have $\|\Delta b(q)\| \leq \|\Delta\|\|G(q)\|$.  For
simplicity, we assume that $\Delta$ is a diagonal matrix given by
$\Delta  ={\rm diag}\{\Delta_1, \Delta_2\}$. The robust controller
term $\tau_{r}$ is designed according to (\ref{eq:case2}), where
$H=A^{-1},\;L=G$, and finally the two unknown parameters
$\Delta_1$ and $\Delta_2$ are estimated by the MES, as shown in
the next section.

\subsection{MES Based uncertainties estimation}
\label{sec:constant}

First, we choose the following performance cost function
\begin{equation}\begin{array}{c} J=Q_1\int_{0}^{t_f}
(q-q_d)^T(q-q_d) {\rm{d}} t \\+Q_2\int_{0}^{t_f} (\dot q-\dot
{q}_d)^T(\dot q-\dot q_d) \rm{d}t,
\end{array}\label{eq:cost}
\end{equation}
where $Q_1>0$ and $Q_2>0$ denote the weighting parameters. Then,
the two unknown parameters $\Delta_1$ and $\Delta_2$ are estimated
by the MES (which is a discrete version of (\ref{eq:mes}))
\begin{equation}
\begin{array}{l}
\displaystyle {x}_i (k+1) ~=~x_i(k) + a_i t_f \sin(\omega_i t_f  k +\frac{\pi}{2}) J\\
\displaystyle  \widehat{\Delta}_i(k+1)~=~x_i(k+1) + a_i
\sin(\omega_i  t_f k -\frac{\pi}{2}) ,\quad i  =1,2
\end{array}
\label{eq:esa}
\end{equation}
where $k ~=~0,~1,~2,\cdots$ denotes the iteration index, $x_i$ and
$\widehat{\Delta}_i$ ($i = 1,2$) start from zero initial
conditions. We simulate the system with $\Delta_1 = -1$ and
$\Delta_2 = -3$. The parameters that were used in the cost
function (\ref{eq:cost}) and the MES (\ref{eq:esa}) are summarized
in Table \ref{tab:case2}.
As shown in Fig. \ref{fig:comp_case2}, the ISS-based controller
combined with ES greatly improves the tracking performance. Fig.
\ref{fig:cost_case2} shows that the cost function starts at an
initial value around $6$ and decreases below $0.5$ within $100$
iterations and the value of the cost function is decreasing over
the iterations. Moreover, the estimate of the unknown parameters
converge to a neighborhood of the true parameter values, as shown
in Fig. \ref{fig:estimate_case2}.

\begin{table}
\caption{Parameters used in MES} \label{tab:case2} \centering
\begin{tabular}{ |c|c|c|c|c|c|c| }
 \hline
$Q_1$ & $Q_2$ & $a_1$ & $a_2$& $\omega_1$ &$\omega_2$& $t_f$ \\ [0.5ex]
\hline
$5$ & $5$ & $0.05$ & $0.04$ & $7.4$ & $7.5$ & $4$  \\
 \hline
\end{tabular}
\end{table}
%
\begin{figure}
\vspace{-0cm}
 \begin{center}
 \includegraphics[scale=0.45]{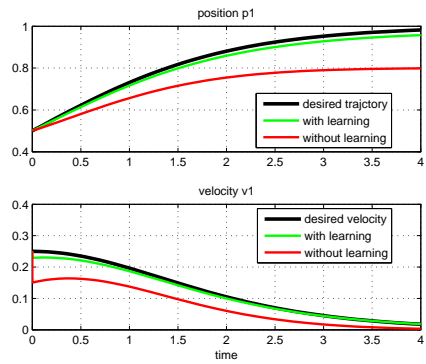}
 \vspace{-0cm}
 \caption{Obtained trajectories vs. Reference Trajectories for $q_1$ [rad] (in the top plot) and $\dot{q}_1$ [rad/sec] (in the bottom plot)}
 \end{center}
 \label{fig:comp_case2}
\end{figure}
\begin{figure}
 \centering
 \vspace{-0cm}
 \hspace{-0cm}\includegraphics[scale=0.45]{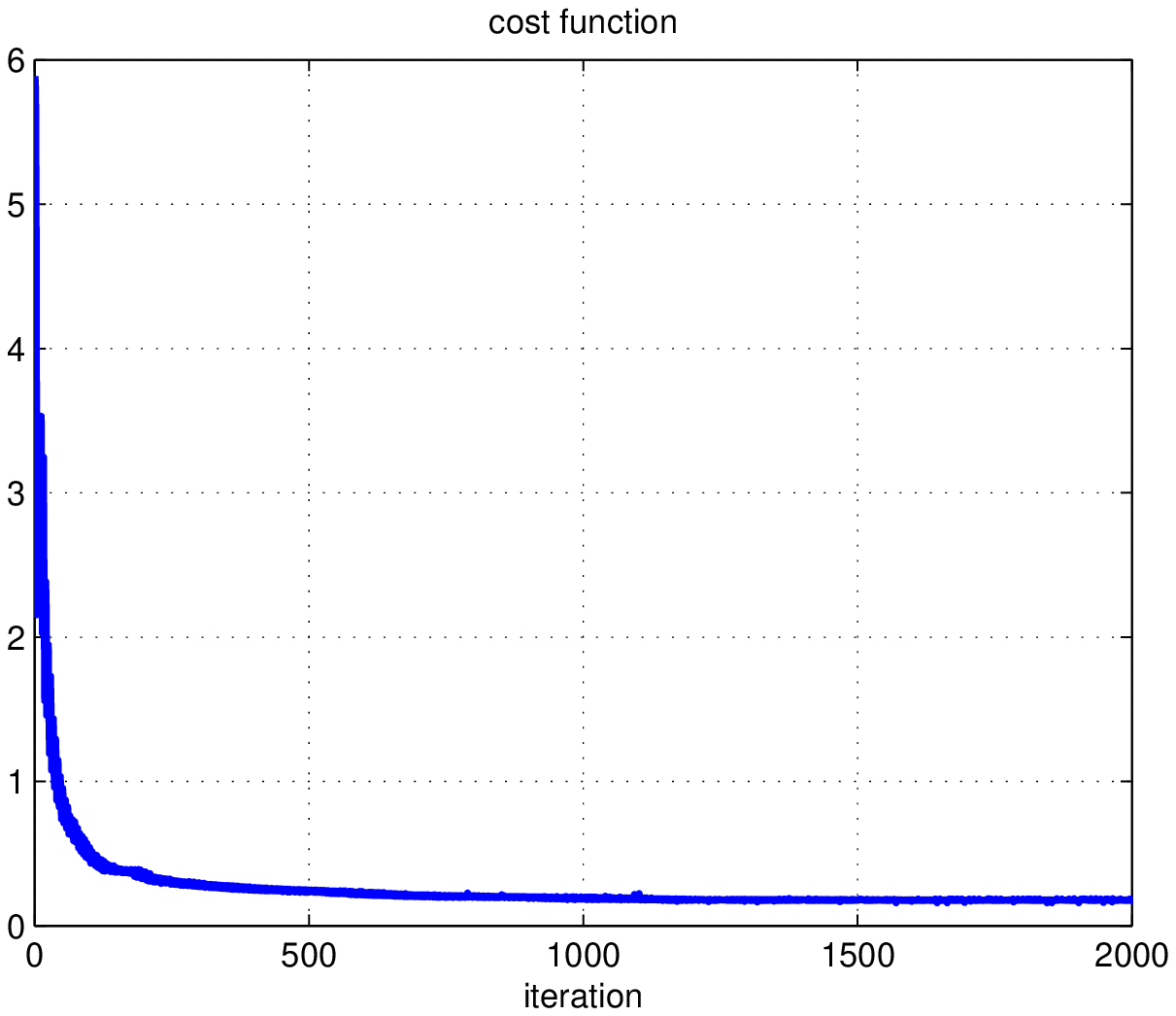}
 \vspace{-0cm}
 \caption{The cost function vs. the number of iterations}
 \label{fig:cost_case2}
\end{figure}
\begin{figure}
 \centering
 \vspace{-0cm}
 \includegraphics[scale=0.45]{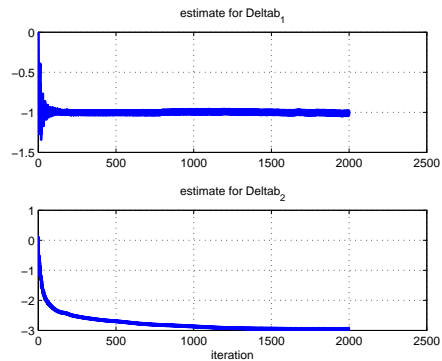}
 \vspace{-0cm}
 \caption{Estimate of parameter $\Delta_1$ (in the top plot) and parameter $\Delta_2$ (in the bottom plot)}
 \label{fig:estimate_case2}
\end{figure}
%
%
%
%
%
%
%
%
%
\subsection{Time dependent uncertainties}
Now we assume that there exist additive time-dependent
uncertainties in the model (\ref{eq:norm_robot}), i.e.
\begin{align}
\label{eq:uncertain_robot} \ddot{q}=~H^{-1}(q) \tau -  H^{-1}(q)
\left[C(q,\dot{q})\dot{q}+ G(q)\right] + \Delta b(q,t).
\end{align}
We will illustrate our approach for the uncertain model
(\ref{eq:uncertain_robot}). In the following, we consider the cost
function
\begin{align}
\label{eq:cost} J~=~Q_1\int_{0}^{t_f} \|q-q_d\|^2 {\rm{d}} t
+Q_2\int_{0}^{t_f} \|\dot q-\dot {q}_d\|^2 \rm{d}t,
\end{align}
where $Q_1>0$ and $Q_2>0$ denote the weighting parameters.

\subsection{Time-varying MES Based Adaptation}
\label{sec:constant} Due to space limitation, we report hereafter
only the case presented in Section \ref{sec:case1}, when $\Delta
b(q,t)$ is simply a time-varying vector $\Delta(t) $, where
$\Delta(t)  =~[\Delta_1(t), \Delta_2(t)]^T$. However, we underline
that we have successfully tested the remaining cases and that all
the results will be reported in a longer journal version of
this work.\\
Here the controller is designed according to Theorem
\ref{thm:case1} and the two unknowns $\Delta_1(t)$ and
$\Delta_2(t)$ are identified by the MES (\ref{eq:esa_gen}) such
that the cost function $J$ in (\ref{eq:cost}) is minimized.

We simulate the system with
\begin{align}
\label{eq:time-varying}
 {\Delta }_1(t)~&=~1 - 0.14\sin(0.01t)\notag\\
 {\Delta }_2(t)~&=~1 - 0.12\cos(0.01t).
\end{align}
The estimate of these two parameters $\widehat{\Delta}_i$ ($i =
1,2$) are computed using a discrete version of (\ref{eq:esa_gen}),
given by
\begin{align}
{\Delta}_i (k+1) ~=&~\Delta_i(k) + t_f(\alpha_i\sqrt{\omega_i}\cos(\omega_i t_f k)\notag\\
&~~-\kappa_i \sqrt{\omega_i} \sin(\omega_i t_f  k) J), \quad i
=1,2 \label{eq:esa}
\end{align}
where $k ~=~0,1,2,\cdots$ denotes the iteration index and
$\widehat{\Delta}_i$ ($i = 1,2$) start from zero initial
conditions. The parameters used in the cost function
(\ref{eq:cost}) and the MES (\ref{eq:esa}) are summarized in Table
\ref{tab:varying1}. For more details about the tuning of the
parameters in the MES, we refer the reader to \cite{Ariy03}.
However, we underline here that the frequencies $\omega_1$ and
$\omega_2$ have been selected high enough to ensure efficient
exploration on the search space and to ensure convergence and that
the amplitudes $\alpha_i$ and $\kappa_i$ ($i = 1,2$) of the dither
signals have been selected such that the search is fast enough.
The uncertain parameter vector $\widehat{\Delta}$ is updated for
each cycle, i.e. at the end of each cycle at $t = t_f$, the cost
function $J$ is updated, and the new estimate of the parameters is
computed for the next cycle. The purpose of using MES along with
the ISS-based controller is to improve the performance of the
controller by better estimating the unknown parameters over many
cycles, hence decreasing the estimation errors over time to
provide better trajectory tracking. As shown in Fig.
\ref{fig:comp_case1}, the ISS-based controller combined with ES
greatly improves the tracking performance. Fig.
\ref{fig:costzoom_case1} show that the cost function starts at an
initial value around $7$ and decreases below $1$ within $20$
iterations. Moreover, the estimate of the unknown parameters
converge to a neighborhood of the true parameter trajectories
within $100$ iterations, as shown in Fig.
\ref{fig:estimate_case1}.


\begin{table}
\caption{parameters used for case 1.} \label{tab:varying1}
\centering
\begin{tabular}{ |c|c|c|c|c|c|c|c|c|c| }
 \hline
$Q_1$ & $Q_2$ & $\alpha_1$ & $\alpha_2$& $\omega_1$ &$\omega_2$&
$t_f$ &$\kappa_1$ &$\kappa_2$ \\ [0.5em] \hline
$0.325$ & $0.325$ & $0.01$ & $0.01$ & $9.9$ & $9.8$ & $4$  & $0.01$ & $0.01$\\
 \hline
\end{tabular}
\end{table}


\begin{figure}[!t]
 \centering
 \includegraphics[scale=0.45]{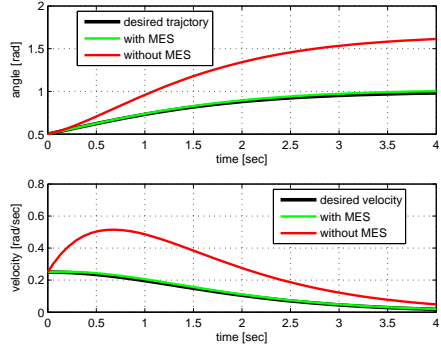}
 \caption{Case 1: Obtained trajectories vs. reference trajectories for $q_1$ (in the top plot) and $\dot{q}_1$ (in the bottom plot).}
 \label{fig:comp_case1}
\end{figure}

\begin{figure}[!t]
 \centering
 \includegraphics[scale=0.45]{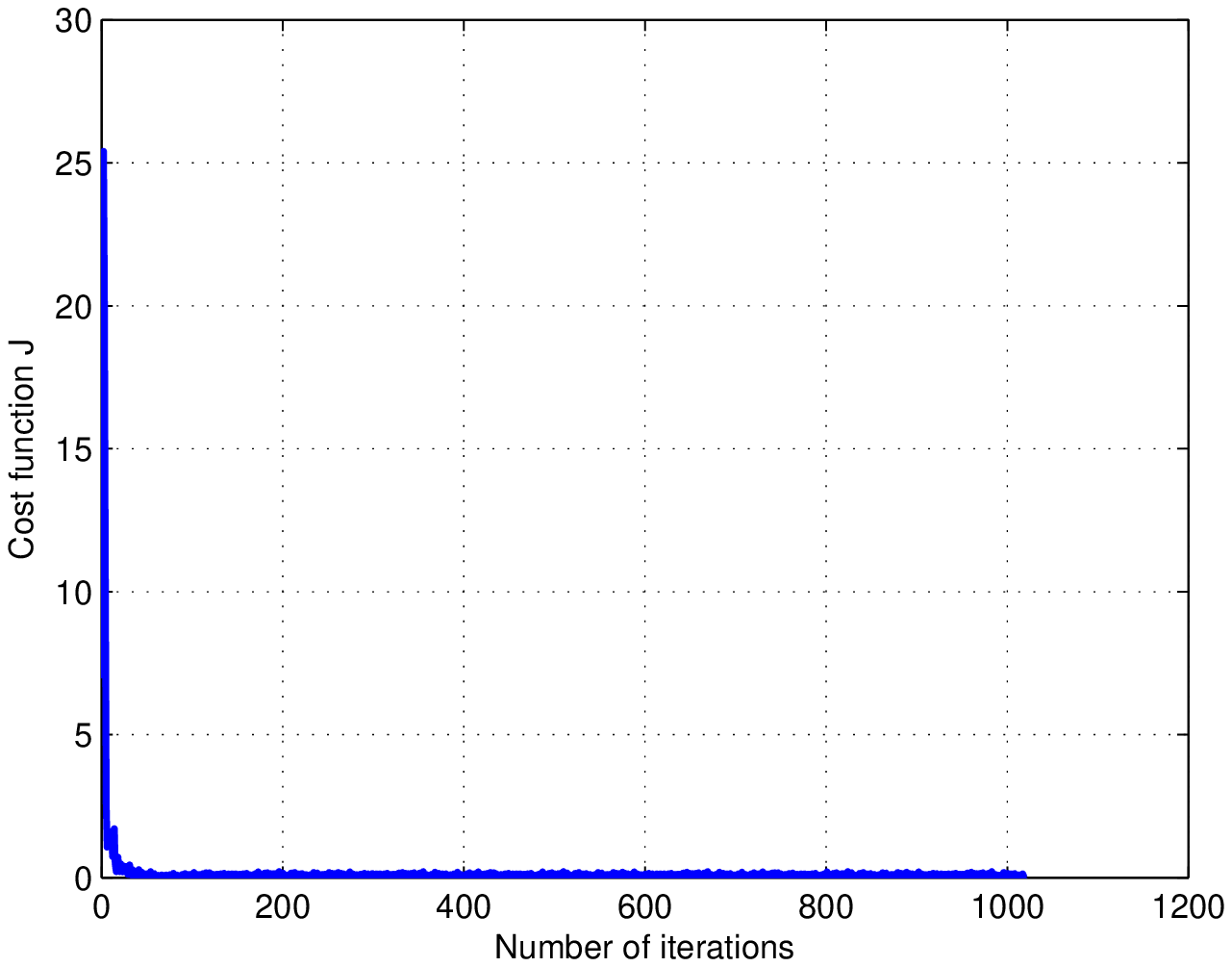}
 \caption{Case 1: The cost function vs. the number of iterations.}
 \label{fig:cost_case1}
\end{figure}

\begin{figure}[!t]
 \centering
 \includegraphics[scale=0.45]{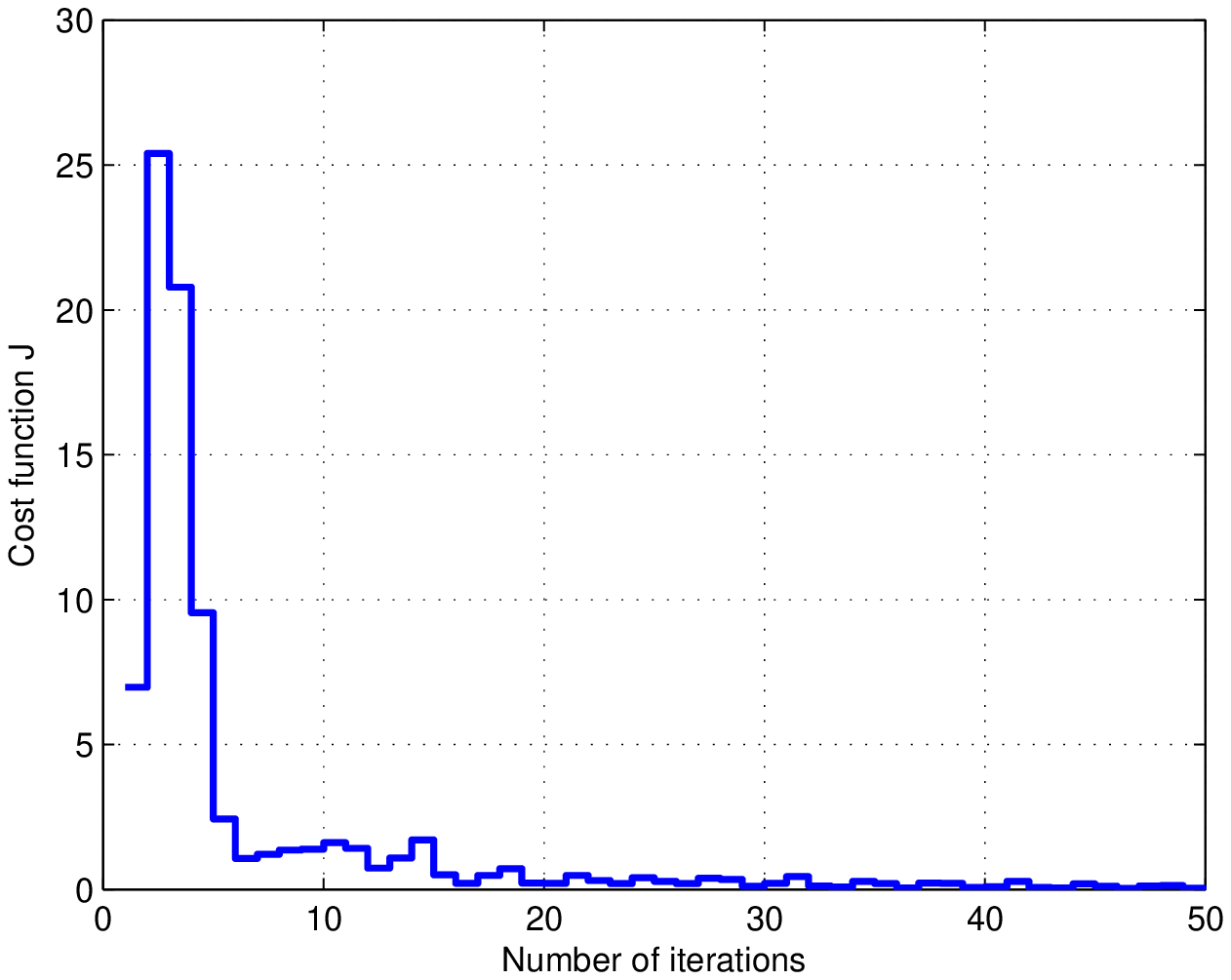}
 \caption{Case 1: The zoom-in cost function vs. the number of iterations.}
 \label{fig:costzoom_case1}
\end{figure}

\begin{figure}[!t]
 \centering
 \includegraphics[scale=0.45]{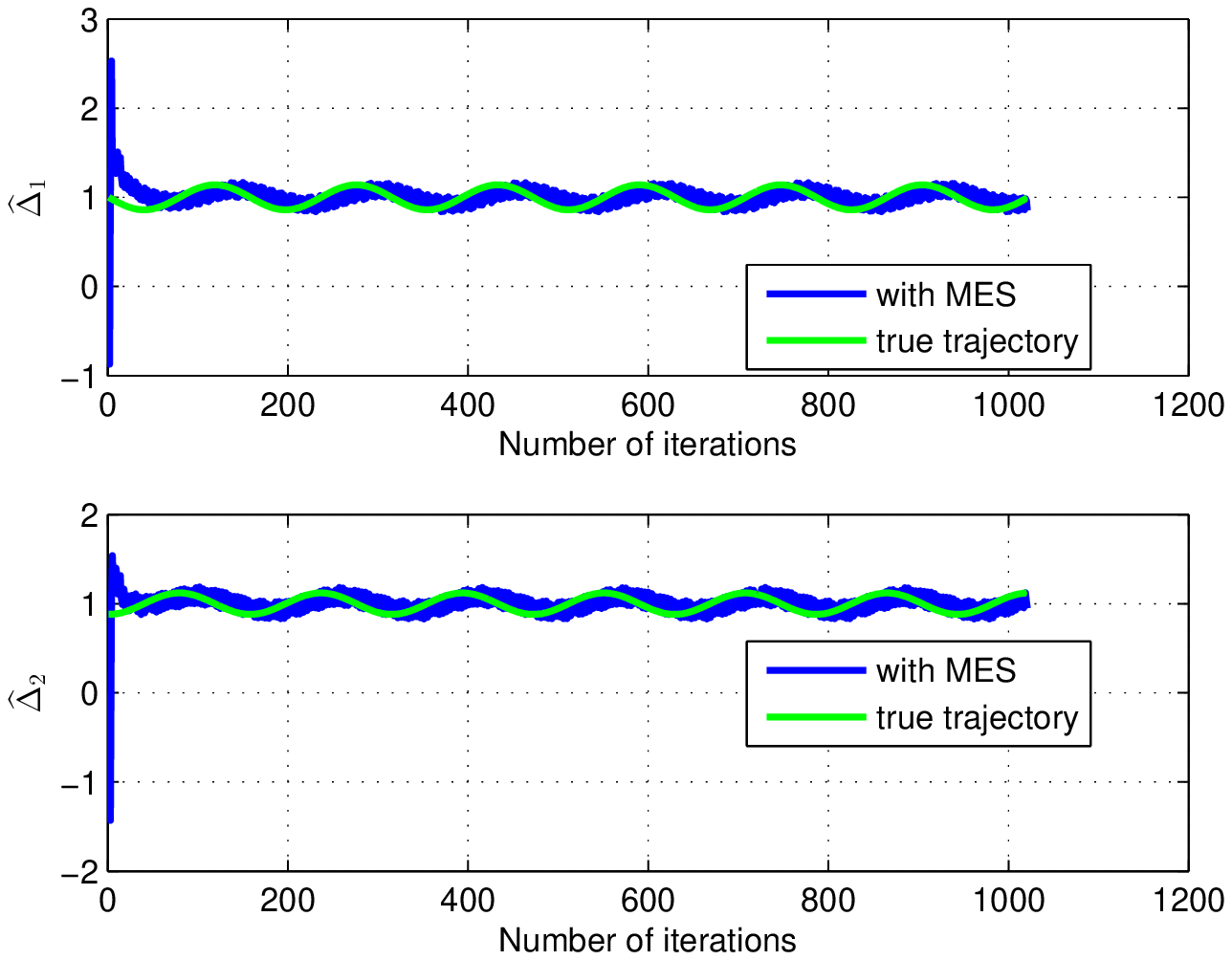}
 \caption{Case 1: Estimate of parameter $\Delta_1$ (in the top plot) and parameter $\Delta_2$ (in the bottom plot).}
 \label{fig:estimate_case1}
\end{figure}

\section{Conclusion}
\label{sec:con} In this paper, we studied the problem of extremum
seeking-based indirect adaptive control for nonlinear systems
affine in the control with bounded additive state-dependent
uncertainties. We have proposed a robust controller which renders
the feedback dynamics ISS w.r.t the parameter estimation errors.
Then we have combined the ISS feedback controller with a
model-free ES algorithm to obtain a learning-based adaptive
controller, where the ES is used to estimate the uncertain part of
the model. We have presented the stability proof of this
controller and have shown a detailed application of this approach
on a two-link robot manipulator example. Future work will deal
with considering controllers under input constraints, using
different ES/learning algorithms with less restrictive tuning
conditions, and comparing the obtained controllers to some
available classical nonlinear adaptive controllers.

\bibliographystyle{IEEEtran}
\bibliography{mxia_paper}

\end{document}